\pdfoutput=1
\documentclass[letterpaper,11pt]{article}
\usepackage{style}
\usepackage{shortcuts}

\title{Beating Bellman's Algorithm for Subset Sum}
\author{Karl Bringmann\thanks{Saarland University and Max Plank Institute for Informatics. This work is part of the project TIPEA that has received funding from the European Research Council (ERC) under the European Unions Horizon 2020 research and innovation programme (grant agreement No. 850979).} \and Nick Fischer\thanks{INSAIT, Sofia University ``St. Kliment Ohridski''. This research was partially funded from the Ministry of Education and Science of Bulgaria (support for INSAIT, part of the Bulgarian National Roadmap for Research Infrastructure). Parts of this work were done while the author was at Weizmann Institute of Science.} \and Vasileios Nakos\thanks{National and Kapodistrian University of Athens and Archimedes / Athena RC.}}
\date{}

\begin{document}

\maketitle

\begin{abstract}
\noindent
Bellman's algorithm for Subset Sum is one of the earliest and simplest examples of dynamic programming, dating back to 1957. For a given set of $n$ integers $X$ and a target $t$, it computes the set of \emph{subset sums} $\mathcal S(X, t)$ (i.e., the set of integers $s \in [0\ldots t]$ for which there is a subset of~$X$ summing to $s$) in time $\Order(|\mathcal S(X, t)| \cdot n)$. Since then, it has been an important question whether Bellman's seminal algorithm can be improved.

This question is addressed in many recent works. And yet, while some algorithms improve upon Bellman's algorithm in specific parameter regimes, such as Bringmann's~\smash{$\widetilde O(t + n)$}-time algorithm [SODA~'17] and Bringmann and Nakos'~\smash{$\widetilde O(|\mathcal S(X, t)|^{4/3})$}-time algorithm [STOC~'20], none of the known algorithms beats Bellman's algorithm in all regimes. In particular, it remained open whether Subset Sum is in time~\smash{$\widetilde O(|\SSS(X, t)| \cdot n^{1-\epsilon})$} (for some $\epsilon > 0$).

In this work we positively resolve this question and design an algorithm that outperforms Bellman's algorithm in all regimes. Our algorithm runs in time $\widetilde O(|\mathcal S(X, t)| \cdot \sqrt{n})$, thus improving the time complexity by a factor of nearly $\sqrt n$. Our key innovation is the use of a result from additive combinatorics, which has not been applied in an algorithmic context before and which we believe to be of further independent interest for algorithm design. To demonstrate the broader applicability of our approach, we extend our ideas to a variant of Subset Sum on vectors as well as to Unbounded Subset Sum.
\end{abstract}

\setcounter{page}{0}
\thispagestyle{empty}
\clearpage

\section{Introduction} \label{sec:intro}
In this paper we study the seminal \emph{Subset Sum} problem defined as follows. Consider a set of~$n$ \emph{items}~\smash{$X \subseteq \Int_{\geq 0}$}. We call an integer $s$ a \emph{subset sum} of $X$ if there is a subset of~$X$ that sums to $s$, and we denote the set of all subset sums by $\SSS(X)$. The \emph{Subset Sum} problem is, given~$X$ and a \emph{target}~\makebox{$t \in \Int_{\geq 0}$}, to decide whether $t$ is a subset sum of $X$ (i.e., whether $t \in \SSS(X)$). Subset Sum is arguably the most fundamental NP-hard problem at the intersection of computer science, mathematical optimization, and operations research, and it is among the \emph{simplest} NP-hard problems, being a critical special case of the famous Knapsack and Integer Programming problems. Subset Sum also plays an important role in cryptography (with close connections to lattice-based cryptography) and from a purely mathematical viewpoint (specifically, in additive combinatorics it dates back to 1938 as the so-called \emph{Littlewood-Offord} problem).

\begin{algorithm}[t]
\caption{Bellman's dynamic programming algorithm to compute the set of subset sums $\SSS(X, t)$} \label{alg:bellman}
\begin{algorithmic}[1]
\Procedure{Bellman}{$X, t$}
    \State $S \gets \set{0}$
    \ForEach{$x \in X$}
        \State $S \gets \set{s, s + x : s \in S} \cap [0 \ldots t]$
    \EndForEach
    \State\Return $S$
\EndProcedure
\end{algorithmic}
\end{algorithm}

The rich algorithmic landscape of Subset Sum can be divided into two lines of research, based on two fundamental algorithms that are frequently taught in basic computer science classes. In one line, the textbook meet-in-the-middle algorithm~\cite{HorowitzS74} solves Subset Sum with an exponential time complexity of $\Order(2^{n/2} \poly(n))$. Despite extensive efforts~\cite{SchroeppelS81,Howgrave-GrahamJ10,DinurDKS12,AustrinKKM13,AustrinKKN15,AustrinKKN16,BansalGNV18,NederlofW21}, it remains unanswered whether this running time can be improved. Central to this paper is the other line, based on the classic \emph{dynamic programming} algorithm running in pseudo-polynomial time. This algorithm, introduced by Bellman in his foundational book ``Dynamic Programming'' from 1957~\cite{Bellman57}, computes the set of all subset sums~\makebox{$\SSS(X, t) := \SSS(X) \cap [0 \ldots t]$} with a time complexity of $\Order(|\SSS(X, t)| \cdot n)$, which is often bounded by $\Order(t n)$; see \cref{alg:bellman} for the pseudocode. Bellman's algorithm constitutes one of the first and simplest examples of dynamic programming, and forms the baseline for numerous follow-up results~\cite{Bellman57,Pisinger99,Pisinger03,KoiliarisX19,Bringmann17,KoiliarisX18,JinW19,BringmannN20,BringmannW21,ChenLMZ24a,ChenLMZ24b}. Despite this extensive research, it is not known whether Bellman's running time can be improved by even a moderate factor of $n^{\epsilon}$ in the worst case, i.e., the following question remained open:
\begin{quote}
    \centering
    \medskip
    \emph{Can Bellman's algorithm be improved to time $\widetilde\Order(|\SSS(X, t)| \cdot n^{1-\epsilon})$ for some $\epsilon > 0$?}
\end{quote}

\subsection{State of the Art}
While this question indeed has not been answered in its full generality, there has been increasing effort on what can be viewed as improving Bellman's algorithm in special cases. A successful line of research studied Subset Sum with the target $t$ as the parameter; in this parameterization Bellman's algorithm runs in time $\Order(t n)$. The first mild improvement is due to Pisinger~\cite{Pisinger03} who achieved a log-shave in the word RAM model. Then, following Koiliaris and Xu's~\cite{KoiliarisX19,KoiliarisX18} deterministic $\widetilde\Order(t \sqrt{n})$ algorithm, Bringmann~\cite{Bringmann17} devised a randomized algorithm running in time~\makebox{$\widetilde\Order(t + n)$}. See also~\cite{JinW19} for an alternative simpler proof. This running time is optimal in the sense that fine-grained lower bounds rule out algorithms with dependence~$t^{0.99}$~\makebox{\cite{Bringmann17,AbboudBHS22}}. However, these algorithms~\cite{Bringmann17,JinW19} only outperform Bellman's algorithm in the \emph{dense} regime where~\makebox{$|\SSS(X, t)| \approx t$} (specifically, for instances where $|\SSS(X, t)| \gg t / n$).

On the other hand, focusing on the ``output-sensitive'' setting, Bringmann and Nakos~\cite{BringmannN20} designed an $\widetilde\Order(|\SSS(X, t)|^{4/3})$-time algorithm. This again is preferable over Bellman's algorithm in a specific \emph{sparse} regime, namely when $|\SSS(X, t)| \ll n^3$. (Notably, just like Bellman's algorithm all aforementioned algorithms compute the entire set $\SSS(X, t)$.) Additionally, Subset Sum has been studied with respect to many other natural parameters such as the maximum item size~\cite{Pisinger99,BringmannW21,PolakRW21,ChenLMZ24a,ChenLMZ24b,AxiotisBJTW19,AxiotisBBJNTW21,CardinalI21,Potepa21}.

In summary, none of the known algorithms addresses the intermediate regime between extremely dense and extremely sparse and thus none of the known algorithms unconditionally beats Bellman's algorithm. On the other hand, we have so far no reason to suspect that time~\smash{$\widetilde\Order(|\SSS(X, t)| \cdot n)$} is best-possible. The only barrier we are aware of was introduced by Bringmann and Nakos~\cite{BringmannN20}: They rule out truly linear time $\widetilde\Order(|\SSS(X, t)|)$ for a restricted but natural class of algorithms based on ``rectangle coverings''.

\subsection{Our Contribution}
The core contribution of this paper is a positive resolution of our driving question, marking the first true improvement on Bellman's algorithm since over 60 years. Specifically, we show the following:

\begin{theorem} \label{thm:subset-sum}
Given a multiset $X \subseteq [0 \ldots t]$, we can compute $\SSS(X, t)$ in time $\widetilde\Order(|\SSS(X, t)| \cdot \sqrt{n})$ (by a randomized algorithm that succeeds with high probability).
\end{theorem}

Our algorithm is based on a careful combination of three main tools: The first is Bringmann's color coding technique~\cite{Bringmann17}. The second is sparse convolution to compute sumsets in output-sensitive time. The third and most exciting ingredient is an estimate from additive combinatorics required for the key step in the analysis. Specifically, we rely on a certain \emph{sub-multiplicativity} property of sumsets proved by Gyarmati, Matolcsi and Ruzsa~\cite{GyarmatiMR10}. It states that for all sets $A_1, \dots, A_K$ we can bound the size of the sumset
\begin{equation*}
    |A_1 + \dots + A_K| \leq \parens*{\prod_{i=1}^K |B_i|}^{\frac{1}{K-1}},
\end{equation*}
where $B_i = A_1 + \dots + A_{i-1} + A_{i+1} + \dots + A_K$ is the same sumset where we omit one term (see \cref{thm:submultiplicativity} for details). All in all, we arrive at a satisfyingly clean and even somewhat simple algorithm.

We remark that the same ``three-component-recipe'' has also been successfully applied in the recent algorithms due to Chen, Lian, Mao and Zhang~\cite{ChenLMZ24c,ChenLMZ24d}. More generally, employing tools from additive combinatorics in fine-grained algorithms for Subset Sum-type problems has evolved into a modern trend~\cite{ChanL15,BringmannN20,BringmannW21,BringmannN21,AbboudBF23,JinX23,ChenLMZ24a,Bringmann24,Jin24,ChenLMZ24b,ChenLMZ24c,ChenLMZ24d}. However, the specific additive combinatorics result we rely on has not been applied in algorithm design before (to the best of our knowledge). We thus bring a new additive combinatorics tool to the table of algorithm design, and we are confident that there are further applications of this tool---for Subset Sum or possibly beyond---given the general prevalence of sumsets in algorithms.

Of course it is desirable to improve our result further, to time $\widetilde\Order(|\SSS(X, t)| \cdot n^{1/3})$, say, or perhaps even to $\widetilde\Order(|\SSS(X, t)|)$. We leave this as an open question. The bottleneck of our algorithm is a routine called \emph{prefix-restricted sumset} computations as first considered by~\cite{BringmannN20}, for which the currently known algorithms have an overhead of $\sqrt{n}$ (for details see \cref{sec:subset-sum} and \cref{thm:prefix-sumset-high-dim}). Unfortunately, all known algorithms further fall into the class of ``rectangle covering'' algorithms mentioned before, and it thus unconditionally ruled out that this class of prefix-restricted sumset algorithms leads to near-linear running time~\cite{BringmannN20}. It is our impression that very different ideas are necessary to make further progress.

\subsection{Extensions}
To explore the flexibility of our new approach we investigate whether we can achieve similar improvements for other Subset Sum-type problems.

\paragraph{Extension 1: Unbounded Subset Sum}
We call an integer an \emph{unbounded} subset sum of $X$ if it can be expressed as the sum of elements from $X$ that is allowed to contain each element several times (and not only once). We write $\SSS^*(X, t)$ for the set of unbounded subset sums of $X$ in the range~$[0 \ldots t]$. Computing the unbounded subset sums easily reduces to the bounded case by duplicating each item $x$ to items $x, 2x, 4x, \dots \leq t$. Typically, however, Unbounded Subset Sum has significantly simpler and \emph{deterministic} algorithms (as we can avoid color coding). We show that our technique follows this heuristic and extends to Unbounded Subset Sum with essentially the same time complexity:

\begin{theorem}[Unbounded Subset Sum] \label{thm:unbounded}
Given a multiset $X \subseteq [0 \ldots t]$ of size $n$ we can compute $\SSS^*(X, t)$ in time~\smash{$\Order(|\SSS^*(X, t)| \cdot \sqrt{n} \cdot \polylog(t))$} (by a deterministic algorithm).
\end{theorem}

Here the additional overhead $\polylog(t)$ is mainly caused by the deterministic sparse convolution algorithm~\cite{BringmannFN22}. 

\paragraph{Extension 2: High-Dimensional Subset Sum}
We also consider the natural generalization from Subset Sum on integers to high dimensions, i.e., vectors. That is, for a set of integer vectors~\smash{$X \subseteq [0 \ldots t]^d$}, the goal is to compute the set of subset sums in the cube $[0 \ldots t]^d$, denoted as before by $\SSS(X, t)$. We assume $d$ to be constant. This version of Subset Sum can still be solved in time~$\Order(|\SSS(X, t)| \cdot n)$ by Bellman's algorithm (without changes). It turns out that our algorithm generalizes in an interesting way, with weaker bounds for higher dimensions, but nevertheless a polynomial improvement over Bellman's algorithm:

\begin{restatable}[High-Dimensional Subset Sum]{theorem}{thmsubsetsumhighdim} \label{thm:subset-sum-high-dim}
Let $d \geq 1$ be a constant. Given a size-$n$ multiset~\makebox{$X \subseteq [0 \ldots t]^d$} we can compute $\SSS(X, t)$ in time~\smash{$\widetilde\Order(|\SSS(X, t)| \cdot n^{1-\frac{1}{d+1}})$} (by a randomized algorithm that succeeds with high probability).
\end{restatable}

For Unbounded Subset Sum the above result can be achieved deterministically with a $\polylog(t)$ factor overhead, see section~\ref{sec:unbounded} for details. Similarly to the one-dimensional case, the bottleneck of the high-dimensional variants lies in prefix-restricted computations (in high dimensions).

\subsection{Outline}
We structure the remainder of this paper as follows. In \cref{sec:prelims} we clarify some preliminary definitions. Then we describe our main algorithm for Subset Sum in \cref{sec:subset-sum}; for the convenience of the reader we start with a high-level technical overview of our ideas. In the following \cref{sec:high-dim,sec:unbounded} we then provide details for the two extensions.
\section{Preliminaries} \label{sec:prelims}
For real numbers $m, n$, we write $[m \ldots n] = \set{ \lceil m \rceil, \dots, \lfloor n \rfloor}$ and $[n] = [1 \ldots n]$. We also set $\poly(n) = n^{\Order(1)}$, $\polylog(n) = (\log n)^{\Order(1)}$ and~\smash{$\widetilde\Order(n) = n \polylog(n)$}.

\paragraph{Sumsets}
For sets $A, B \subseteq \Int^d$ we define the \emph{sumset} $A + B = \set{a + b : a \in A,\, b \in B}$. For a single integer $a \in \Int$ we occasionally write $a + B = \set{a + b : b \in B}$. The following well-known lemma~\cite[Lemma~5.3]{TaoV06} will be useful throughout; we include a quick proof for convenience.

\begin{lemma} \label{lem:sumset-lower-bound}
Let $A_1, \dots, A_K \subseteq \Int$. Then 
\begin{equation*}
    |A_1|+ \dots + |A_K| \le |A_1+ \dots + A_K| + K - 1.
\end{equation*}
\end{lemma}
\begin{proof}
It suffices to prove the statement for $K = 2$; the full statement then easily follows by induction. To show that $|A + B| \geq |A| + |B| - 1$, let $a_1 < \dots < a_{|A|}$ and $b_1 < \dots < b_{|B|}$ denote the elements in $A$ and $B$, respectively. Then note that~\makebox{$a_1 + b_1 < \dots < a_1 + b_{|B|} < a_2 + b_{|B|} < \dots < a_{|A|} + b_{|B|}$}, and all of these $|A| + |B| - 1$ elements are clearly contained in $A + B$.
\end{proof}

It is the topic of a long line of research that sumsets can be computed in output-sensitive time~\cite{ColeH02,ChanL15,ArnoldR15,Nakos20,GiorgiGC20,BringmannFN21,BringmannFN22,JinX24}. Specifically, we will rely on the following facts:

\begin{lemma}[Sumset Computation~\cite{BringmannFN22,JinX24}] \label{lem:sparse-conv}
Given $A, B \subseteq [0\ldots t]$ we can compute $A+B$ in time
\begin{itemize}
    \item $|A+B|\cdot \polylog(t)$ by a deterministic algorithm, and
    \item $\Order(|A+B|\cdot \log|A+B|)$ by a Las Vegas algorithm (assuming a word RAM model with word size $\Theta(\log t)$).
\end{itemize} 
\end{lemma}

From this lemma it follows also that we can compute sumsets of vector sets $A, B \subseteq [0\dots t]^d$ with an overhead of $\poly(d)$.\footnote{Specifically, associate each vector $x \in [0\ldots t]^d$ with its base-$(2t+1)$-encoded integer $\phi(x) = \sum_{i=1}^d x[i] (2t+1)^{i-1}$. Then encode $A$ and $B$ as $\phi(A) = \set{\phi(a) : a \in A}$ and $\phi(B)$, compute the integer sumset $\phi(A) + \phi(B)$ and translate back into vectors. The integers have size at most $(2t + 1)^d$, and thus to simulate the sumset algorithm (which assumes a word RAM with word size $\Theta(d \log t)$) on a word RAM with word size $\Theta(\log t)$ we incur an overhead of $\poly(d)$.}

\paragraph{Subset Sum}
To define the (high-dimensional) Subset Sum problem, let $X = \set{x_1, \dots, x_n} \subseteq \Int_{\geq 0}^d$ be a multiset. We define the set of \emph{subset sums} $\SSS(X) = \set{0, x_1} + \dots + \set{0, x_n}$, and the (infinite) set of \emph{unbounded subset sums} $\SSS^*(X) = \set{0, x_1, 2x_1, \dots} + \dots + \set{0, x_n, 2 x_n, \dots}$. For a target $t$, we set~\makebox{$\SSS(X, t) = \SSS(X) \cap [0\ldots t]^d$} and $\SSS^*(X, t) = \SSS^*(X) \cap [0\ldots t]^d$. Then the Subset Sum problem is to compute $\SSS(X, t)$ and the Unbounded Subset Sum problem is to compute $\SSS^*(X, t)$, or, as decision versions, to decide whether $t$ is an (unbounded) subset sum.

We throughout assume that $|X| \leq |\SSS(X, t)|$. This assumption is always true for sets $X \subseteq [0\ldots t]^d$, and can also easily be enforced for multisets in a linear-time preprocessing to remove redundant elements.

\paragraph{Randomization}
We say that an event happens \emph{with high probability} if it happens with probability $1 - n^{-c}$, for some arbitrary prespecified constant $c$. 
\section{Overview for Subset Sum} \label{sec:subset-sum}
In this section we give a high-level description the ideas behind our new Subset Sum algorithm (for the 1-dimensional case, $d = 1$). We defer all formal proofs to \cref{sec:high-dim} (where we prove them in more generality directly for higher dimensions).

\subsection{Overview} \label{sec:subset-sum:sec:overview}
Many algorithms for Subset Sum rely in one way or another on the following general idea: Consider any partition $X = X_1 \sqcup X_2$ (which may be random, or cleverly chosen). If we can efficiently compute the sets $\SSS(X_1, t)$ and $\SSS(X_2, t)$, then we can also compute
\begin{equation} \label{eq:divide-and-conquer}
    \SSS(X, t) = (\SSS(X_1, t) + \SSS(X_2, t)) \cap [0 \dots t].
\end{equation}
This approach is helpful for the design of divide-and-conquer algorithms~\cite{Bringmann17,KoiliarisX19,BringmannN20,ChenLMZ24c,ChenLMZ24d} (i.e, the set~$X$ is repeatedly partitioned into smaller sets), or also to initially partition $X$ into a small number of subsets to ensure some common properties (e.g., partition $X$ into subsets~\makebox{$X \cap (\frac{t}{2}, t],\, X \cap (\frac{t}{4}, \frac{t}{2}],\, \dots$} each of which only contains items of some common scale). We will design a divide-and-conquer algorithm based on the same high-level approach, but which requires significantly more care due to the following issue.

\paragraph{Prefix-Restricted Sumsets}
Computing $\SSS(X, t)$ in~\eqref{eq:divide-and-conquer} is computationally quite nontrivial. To abstract from the above setup let $A = \SSS(X_1, t)$ and $B = \SSS(X_2, t)$. Then computing~$\SSS(X, t)$ can be expressed as the following problem called \emph{prefix-restricted} sumset computation: Given two sets~\makebox{$A, B \subseteq [0 \ldots t]$}, the goal is to efficiently compute~\makebox{$C = (A + B) \cap [0 \ldots t]$}. This problem can trivially be solved in time~$\widetilde\Order(t)$ using the Fast Fourier Transform which is sufficient for many applications~\cite{Bringmann17,JinW19}. However, in our ``output-sensitive'' setting (where the goal is a running time proportional to $\SSS(X, t)$) this would be too expensive. Conceptually, the dream solution would be to find a \emph{(near-)linear-time} algorithm for prefix-restricted sumsets, i.e., an algorithm running in time~\smash{$\widetilde\Order(|A| + |B| + |C|)$}. From this it would easily follow that Subset Sum is in best-possible time~\smash{$\widetilde\Order(|\SSS(X, t)|)$}. However, there are barriers towards achieving such a prefix-restricted sumset algorithm~\cite{BringmannN20}, and we leave this as an important open problem.

On the positive side, we are aware of three nontrivial algorithms: (i) it follows from the sparse convolution literature that we can compute $A + B$ (and thereby $C$) in time~\smash{$\widetilde\Order(|A + B|)$}. While this appears very fast at first glance---after all, $C$ is ``just'' the set~\makebox{$A + B \subseteq [0 \ldots 2t]$} restricted to~$[0 \ldots t]$---it turns out that there are sets $A$ and $B$ for which $|C| \ll |A + B|$.\footnote{For an extreme example, consider the sets $A = \frac{t}{2} + \set{1, 2, \dots, n}$ and $B = \frac{t}{2} + \set{n, 2n, \dots, n^2}$ for some target~\makebox{$t \gg n^2$}. Then $|A| = |B| = n$ and $|A + B| = \Omega(n^2)$, whereas $|C| = 0$.} Bringmann and Nakos~\cite{BringmannN20} proposed two other algorithms for prefix-restricted sumset computations: (ii) a simple one running in time~\smash{$\widetilde\Order(\sqrt{|A| \, |B| \, |C|})$}. And (iii), a more involved algorithm based on Ruzsa's triangle inequality that runs in time~\smash{$\widetilde\Order(|A| + |B| + |C|^{4/3})$}. This last algorithm can be turned into the state-of-the-art~\smash{$\widetilde\Order(|\SSS(X, t)|^{4/3})$}-time algorithm for Subset Sum~\cite{BringmannN20}, but due to the $|C|^{4/3}$ dependence we cannot use it for our goal. The other two algorithms are potentially applicable. 

In summary: Suppose that we attempt to split the given instance into parts $X = X_1 \sqcup X_2$, compute $\SSS(X_1, t_1)$ and $\SSS(X_2, t_2)$ (possibly for different targets $t_1, t_2 \leq t$) and recombine $\SSS(X_1, t_1)$ and $\SSS(X_2, t_2)$ by a prefix-restricted sumset computation. Then, for this to be efficient, we have to make sure that either
\begin{enumerate}[label=(\roman*)]
    \item $|\SSS(X_1, t_1) + \SSS(X_2, t_2)|$ or
    \item $\sqrt{|\SSS(X_1, t_1)| \, |\SSS(X_2, t_2)| \, |\SSS(X, t)|}$
\end{enumerate}
is upper-bounded by $\Order(|\SSS(X, t)| \cdot \sqrt{|X|})$.

\paragraph{Large Items}
The previous paragraph severely restricts our freedom in partitioning $X$. Luckily, one very useful partition can nevertheless still be efficiently dealt with: the partition into \emph{small items} $X_S = \set{x \in X : x \leq \gamma t}$ and \emph{large items} $X_L = X \setminus X_S$ (for some parameter~\makebox{$\gamma > 0$} which we will determine later). The main challenge in the algorithm will be to compute the subset sums attained by small items,~$\SSS(X_S, t)$. Suppose for the moment that we have already computed~$\SSS(X_S, t)$; we prove that based on approach~(ii) we can then also take the large items into account and compute~\makebox{$\SSS(X, t) = (\SSS(X_S, t) + \SSS(X_L, t)) \cap [0 \ldots t]$}. Formally, we prove the following lemma (which should be applied with $Z = \SSS(X_S, t)$):

\begin{lemma}[Large Items] \label{lem:large-items}
Let $\gamma > 0$, let $X_L \subseteq [\gamma t, t]$ be a multiset and let $Z \subseteq [0 \ldots t]$. We can compute $Z' = (Z + \SSS(X_L, t)) \cap [0 \ldots t]$ in time $\widetilde\Order(|Z'| \cdot \sqrt{|X_L|} \cdot \poly(\gamma^{-1}))$ (by a Las Vegas algorithm).
\end{lemma}
\begin{proof}[Proof Sketch]
The idea is to apply the \emph{color coding} technique as pioneered by Bringmann~\cite{Bringmann17} in the context of Subset Sum. First observe that any subset sum in $\SSS(X_L, t)$ can involve at most~$\gamma^{-1}$ many elements. Therefore, any element $z' \in Z'$ can be expressed as $z' = z + x_1 + \dots + x_m$ where~\makebox{$z \in Z$} (i.e., the contribution of the small elements), and $x_1, \dots, x_m \in X_L$ where $m \leq \gamma^{-1}$. Now consider a uniformly random partition $X_L = X_1 \sqcup \dots \sqcup X_K$ into $K = \Theta(\gamma^{-2})$ \emph{buckets}. By a standard balls-into-bins argument it follows that $x_1 + \dots + x_m \in (X_1 \cup \set{0}) + \dots + (X_K \cup \set{0})$ with constant probability. Thus, the computation of $Z'$ essentially amounts to computing the iterated prefix-restricted sumset
\begin{equation*}
    (Z + (X_1 \cup \set{0}) + \dots + (X_K \cup \set{0})) \cap [0 \ldots t].
\end{equation*}
Specifically, to compute this set start with $Z' \gets Z$. Then, for each $k \gets 1, \dots, K$ we use Bringmann and Nakos' algorithm for prefix-restricted sumsets to update $Z' \gets (Z' + (X_k \cup \set{0})) \cap [0 \ldots t]$. Each step runs in time~\smash{$\widetilde\Order(\sqrt{|Z'| \, |Z'| \, |X_k|}) = \widetilde\Order(|Z'| \cdot \sqrt{X_L})$}, and so the claimed running time follows.

Finally, recall that this only reports any fixed element in $(Z + \SSS(X_L, t)) \cap [0 \ldots t]$ with constant probability. We could either repeat this process for $\Theta(\log |Z'|)$ times, or derandomize the process altogether by employing Reed-Solomon codes. We defer the details to the proof of the more general \cref{lem:large-items-high-dim}.
\end{proof}

\paragraph{Small Items}
It remains to deal with the small items, i.e., to compute $\SSS(X_S, t)$. Conceptually, the advantage of the small items is that we can expect that any subset $z \in \SSS(X_S, t)$ can be partitioned into subset sums $z = z_1 + \dots + z_K$ where~\smash{$z_1, \dots, z_K \lesssim \frac{t}{K}$}. This insight inspires a divide-and-conquer algorithm along the following lines: Partition $X_S = X_1 \sqcup \dots \sqcup X_K$ into some~$K$ parts, compute the subset sums $\SSS(X_1, \frac{t}{K}), \dots, \SSS(X_K, \frac{t}{K})$ \emph{recursively}, and then recover $\SSS(X_S, t)$ as the sumset~\smash{$\SSS(X_1, \frac{t}{K}) + \dots + \SSS(X_K, \frac{t}{K})$}. This appears somewhat promising as this sumset can indeed be computed in time~\smash{$\widetilde\Order(|\SSS(X, t)|)$} (by approach~(i) from before).

Unfortunately, it is unrealistic to use the exact target $\frac{t}{K}$ for the recursive calls---in particular, this would require that the subset sum $t$ could be perfectly partitioned into smaller sums of size at most~$\frac{t}{K}$. Hence, it seems unavoidable to incur some slack and instead use the target~$(1 + \epsilon) \cdot \frac{t}{K}$ (for some parameter $\epsilon > 0$) for the recursive calls. Then, by a standard application of Hoeffding's inequality it follows that indeed $\SSS(X_S, t) = (\SSS(X_1, (1 + \epsilon) \cdot \frac{t}{K}) + \dots + \SSS(X_K, (1 + \epsilon) \cdot \frac{t}{K})) \cap [0\ldots t]$ for an appropriate choice of parameters:

\begin{lemma}[Small Items] \label{lem:small-items}
Let $K \geq 1$ and let $X_S \subseteq [0 \ldots \frac{t}{K^5}]$ be a multiset. Let $X_S = X_1 \sqcup \dots \sqcup X_K$ denote a uniformly random partition into $K$ parts. Then with probability at least $1 - \exp(-K) \cdot |\SSS(X_S, t)|$ it holds that 
\begin{equation*}
    \SSS(X_S, t) = \Big( \sum_{i=1}^K \SSS(X_i, (1 + \tfrac{1}{K}) \cdot \tfrac{t}{K}) \Big) \cap [0\ldots t].
\end{equation*}
\end{lemma}

The proof of \cref{lem:small-items} is simple and deferred to (the generalization \cref{lem:small-items-high-dim}) in \cref{sec:high-dim:sec:color-coding}. Our main concern is how to make use of this lemma in the presence of the $(1 + \epsilon) = (1 + \frac{1}{K})$ slack. In order to obtain an efficient algorithm we need to bound the cost of the prefix-restricted sumset computation (plus the cost of the recursive calls, but that is a secondary concern). As outlined before we can follow the approaches~(i) or~(ii). It is easy to construct instances where approach~(ii) takes time $\Omega(|\SSS(X_S, t)|^{3/2})$. Thus, to make approach~(i) efficient we need to control
\begin{equation*}
    \abs*{\sum_{i=1}^K \SSS(X_i, (1 + \tfrac{1}{K}) \cdot \tfrac{t}{K})}.
\end{equation*}
Unfortunately, it is a priori quite unclear how to bound the size of this sumset $\sum_k \SSS(X_k, (1 + \tfrac{1}{K}) \cdot \frac{t}{K})$. It could contain numbers as large as $K \cdot (1 + \frac{1}{K}) \cdot \frac{t}{K} = (1 + \frac{1}{K}) \cdot t$, and we can thus not simply charge its size against the size of $\SSS(X_S, t) \subseteq [0 \ldots  t]$. Can we nevertheless control its size with respect to $|\SSS(X_S, t)|$?

\paragraph{Sub-Multiplicativity of Sumsets}
This is where finally the theorem from additive combinatorics comes into play. Specifically, we rely on the following \emph{sub-multiplicativity} property due to Gyarmati, Matolcsi and Ruzsa~\cite{GyarmatiMR10}:

\begin{theorem}[Sub-Multiplicativity of Sumsets~\cite{GyarmatiMR10}]
\label{thm:submultiplicativity}
Let $A_1, \dots, A_K$ be finite sets in a commutative semigroup, and define
\begin{equation*}
    B_i = \sum_{j \in [K] \setminus \{i\}} A_j,
\end{equation*}
i.e.\ the sumset of all sets excluding $A_i$. Then, the following estimate holds:
\begin{equation*}
    |A_1 + \ldots + A_K| \leq \parens*{\prod_{i=1}^K |B_i|}^{\frac{1}{K-1}}.
\end{equation*}
\end{theorem}

In the context of additive combinatorics this is not the first work which relates the size of $K$-fold sumsets $A_1 + \dots + A_K$ in terms of the size of $(K-1)$-fold sumsets $A_1 + \dots + A_{i-1} + A_{i+1} + \dots + A_K$; see e.g.~also~\cite{Lev96,NathansonR99,GyarmatiHR07,Ruzsa07}. The proof of \cref{thm:submultiplicativity} is not too complicated and is built on a clever charging argument.

Coming back to our Subset Sum algorithm, it turns out that \cref{thm:submultiplicativity} yields exactly the bound we were looking for. Specifically, let
\begin{equation*}
    A_i = \SSS(X_i, (1 + \tfrac{1}{K}) \cdot \tfrac{t}{K})
\end{equation*}
and thus
\begin{equation*}
    B_i = \sum_{j \in [K] \setminus \set{i}} A_j \subseteq \SSS(X, (K-1) \cdot (1 + \tfrac{1}{K}) \cdot \tfrac{t}{K}) \subseteq \SSS(X, t).
\end{equation*}
Then the theorem implies that
\begin{equation*}
    \abs*{\sum_{i=1}^K \SSS(X_i, (1 + \tfrac{1}{K}) \cdot \tfrac{t}{K})} = |A_1 + \dots + A_K| \leq \parens*{\prod_{i=1}^K |B_i|}^{\frac{1}{K-1}} \leq |\SSS(X, t)|^{\frac{K}{K-1}} = |\SSS(X, t)|^{1+\frac{1}{K-1}}.
\end{equation*}
By choosing $K \gg \log |\SSS(X, t)|$, the right-hand side becomes $\Order(|\SSS(X, t)|)$ as planned. This completes the description of the main ideas behind our algorithm.

\paragraph{The Complete Algorithm}
We summarize the previous discussion by the pseudocode in \cref{alg:subset-sum}. One interesting technical detail remains: The algorithm accepts as an additional input an approximation $s$ of the \emph{number} of subset sums, formally satisfying that
\begin{equation*}
    |\SSS(X, t)| \leq s \leq \poly(|\SSS(X, t)|).
\end{equation*}
This approximation is necessary to that we can accurately choose the parameter $K$. We can of course simply choose $s = t$, but this would incur a running time overhead of $\polylog(t)$ rather than $\polylog |\SSS(X, t)|$.

Instead, we can compute an approximation $s$ \emph{recursively}. More precisely, we design a recursive algorithm $\mathcal A$ that computes $\SSS(X, t)$ with high probability (without expecting as input an approximation of $|\SSS(X, t)|$). This algorithm $\mathcal A$ first arbitrarily partitions $X = X_1 \sqcup X_2$ and computes recursively the subset sums $\SSS(X_1, \frac{t}{2})$ and $\SSS(X_2, \frac{t}{2})$. Then, letting $s_1 = |\SSS(X_1, \frac{t}{2})|$ and $s_2 = |\SSS(X_2, \frac{t}{2})|$ it computes $s = |X| \cdot s_1^2 \cdot s_2^2$, calls \textsc{FastSubsetSum} on input $(X, t, s)$ and returns the computed output $\SSS(X, t)$.

The running time due to this additional layer of recursion can be controlled by \cref{lem:sumset-lower-bound}; the interesting part is to verify that $s$ is indeed an approximation as desired. To this end, note that each subset sum $x_1 + \dots + x_\ell \in \SSS(X, t)$ can always be split into a sum of (i) one item from $X$ (take the largest $x_i$) plus (ii) two subset sums in $\SSS(X_1, \frac{t}{2})$ plus (iii) two subset sums in $\SSS(X_2, \frac{t}{2})$. See \cref{sec:high-dim:sec:full} for the details and the formal analysis of (a more general version of) \cref{alg:subset-sum}.

\begin{algorithm}[t]
\caption{Fast algorithm for Subset Sum (see \cref{thm:subset-sum})} \label{alg:subset-sum}
\begin{algorithmic}[1]
    \Procedure{FastSubsetSum}{$X, t, s$}
        \If { $|X| \leq \Order(1)$ or $t \leq \Order(1)$}
            \State Solve the instance by Bellman's algorithm
        \EndIf
        \State Let $Z \gets \set{0}$
        \State Let $K \gets \ceil{100 \log s}$
        \State Let $X_S \gets X \cap [0\ldots \frac{t}{K^5}]$ and $X_L \gets X \setminus X_S$

        \medskip
        \State\emph{(Small items)}
        \State Partition $X_S$ randomly into multisets $X_S = X_1 \sqcup \dots \sqcup X_K$
        \For {$k \in [K]$}
            \State Recursively compute $Z_k \gets \textsc{FastSubsetSum}(X_k, (1+\frac{1}{K}) \cdot \frac{t}{K}, s)$
            \State Compute $Z \gets Z + Z_k$ by \cref{lem:sparse-conv} \label{alg:subset-sum:line:sumset}
        \EndFor

        \medskip
        \State\emph{(Large items)}
        \State\Return $(Z + \mathcal{S}(X_L, t)) \cap [0\ldots t]$ computed by \cref{lem:large-items} (with~\smash{$\gamma = \frac{1}{K^5}$}) \label{alg:subset-sum:line:large}
    \EndProcedure
\end{algorithmic}
\end{algorithm}
\section{High-Dimensional Subset Sum} \label{sec:high-dim}
The goal of this section is to generalize \cref{thm:subset-sum} to the following \cref{thm:subset-sum-high-dim} in higher dimensions, i.e., for sets of integer vectors instead of integers.

\thmsubsetsumhighdim*

Recall that our approach is mainly built on three components: An efficient algorithm for prefix-restricted sumset computations for the large items, a divide-and-conquer scheme based on color coding for the small items, as well as \cref{thm:submultiplicativity} from additive combinatorics for the analysis. \cref{thm:submultiplicativity} immediately applies to any additive group (in fact, for any commutative semigroup). In the following two sections we generalize the first two components.

Throughout this section we treat the dimension $d$ as a constant. It is however easy to check that we hide constant factors of the form $\poly(d)$ and the exponents of $\polylog$ is $O(d)$ everywhere.

\subsection{Prefix-Restricted Sumsets} \label{sec:high-dim:sec:prefix-sumset}
\begin{theorem}[High-Dimensional Prefix-Restricted Sumsets] \label{thm:prefix-sumset-high-dim}
Let $d \geq 1$ be a constant. Given sets $A,B \subseteq [0\ldots t]^d$ we can compute $C = (A + B) \cap [0\ldots t]^d$ in time
\begin{itemize}
    \item \smash{$\Order((|A| + |B| + (|A|\, |B|)^{1 - \frac{1}{d+1}} |C|^{\frac{1}{d+1}}) \cdot (\log t)^{\Order(d)})$} by a deterministic algorithm, and
    \item \smash{$\Order((|A| + |B| + (|A|\, |B|)^{1 - \frac{1}{d+1}} |C|^{\frac{1}{d+1}}) \cdot (\log(|A| \, |B|))^{\Order(d)})$} by a Las Vegas algorithm.
\end{itemize}
\end{theorem}
\begin{proof}
We inductively solve the following generalized problem: For some $0 \leq k \leq d$, the task is to compute $C = (A + B) \cap [0\ldots t]^k \times [0\ldots \infty]^{d-k}$ (i.e., we only cap the points falling outside the size\=/$t$ box in the first $k$ coordinates). Our goal is to achieve time~\smash{$\widetilde\Order(|A| + |B| + (|A|\, |B|)^{1 - \frac{1}{k+1}} |C|^{\frac{1}{k+1}})$} (where we comment on the log-factors later). When $k = 0$, we can clearly compute~$C$ in time~\smash{$\widetilde\Order(|C|)$} by a sparse sumset computation (see \cref{lem:sparse-conv}). In the following let $k \geq 1$. 

In a first preprocessing step we ensure that $|A|, |B| \leq |C|$. This can be implemented in near-linear time $\widetilde\Order(|A| + |B| + |C|)$ by testing for each point $a \in A$ whether there exists a point~$b \in B$ with $a + b \in [0\ldots t]^k \times [0\ldots \infty]^{d-k}$ by an orthogonal range searching data structure. The more important steps follow.

For a set $X \subseteq \Int^d$, let us write $\max_k(X) = \max_{x \in X} x[k]$ and $\min_k(X) = \min_{x \in X} x[k]$. Let $g$ be a parameter. The first step is to partition $A = A_1 \sqcup \dots \sqcup A_g$ and $B = B_1 \sqcup \dots \sqcup B_g$ into at most~$g$ parts such that the following properties hold:
\begin{enumerate}[label=(\arabic*)]
    \item $\max\nolimits_k(A_1) \leq \min\nolimits_k(A_2) \leq \max\nolimits_k(A_2) \leq \dots \leq \min\nolimits_k(A_{g-1}) \leq \max\nolimits_k(A_{g-1}) \leq \min\nolimits_k(A_g)$,
    \item $\max\nolimits_k(B_1) < \min\nolimits_k(B_2) \leq \max\nolimits_k(B_2) < \dots < \min\nolimits_k(B_{g-1}) \leq \max\nolimits_k(B_{g-1}) < \min\nolimits_k(B_g)$,
    \item for all $i \in [g]$ we have $|A_i| \leq \Order(|A| / g)$,
    \item for all $j \in [g]$ at least one of the following properties holds:
    \begin{enumerate}[label=(\roman*)]
        \item $|B_j| \leq \Order(|B| / g)$ (in which case we say that $j$ is \emph{light}), or
        \item $\min_k(B_j) = \max_k(B_j)$ (i.e., all vectors in $B_j$ share the same $k$-th coordinate, in which case we say that $j$ is \emph{heavy}). In this case, it further holds that there is an index~$1 \leq i < g$ such that $\max_k(A_i) + \max_k(B_j) \leq t$ and $\min_k(A_{i+1}) + \min_k(B_j) > t$.
    \end{enumerate}
\end{enumerate}
It is not complicated but somewhat tedious to obtain this partition. First, sort the sets $A$ and~$B$ in increasing order with respect to the $k$\=/th coordinate. Then take $A_1, \dots, A_{g/2}$ to be consecutive blocks of $A$ in that order of length at most~$\ceil{2 |A| / g}$ each. This assignment clearly satisfies Properties~(1) and~(3). For the blocks in $B$ we cannot apply the same construction as this does not necessarily satisfy Property~(2) which requires strict inequalities $\max_k(B_j) < \max_k(B_{j+1})$. Let us call an integer $z$ a \emph{pivot} if there are at least~$2 |B| / g$ vectors~\makebox{$b \in B$} with $b[k] = z$; clearly there can be at most $g / 2$ pivots. To assign the partition of~$B$, we first identify all pivots $z$ (in linear time). For each pivot $z$, we then introduce a (heavy) block~\makebox{$B_j = \set{b \in B : b[k] = z}$}. Additionally, for each pivot $z$ we test if there is an index $i$ such that $\min(A_i) \leq t - z < \max(A_i)$. If it exists, this index $i$ is unique; in this case we subdivide $A_i$ at $t - z$ into two sub-blocks. Since we started from~$g / 2$ blocks and perform at most~$g / 2$ subdivisions, we indeed end up with at most $g$ blocks in~$A$. Together, these two steps assert Property~(4)~(ii). After dealing with all the pivots, we can now partition the remaining elements in $B$ into (light) blocks~\makebox{$B_1 \sqcup \dots \sqcup B_{g/2}$} of size $\Order(|B| / g)$, satisfying Property~(4)~(i). In this step we additionally make sure that for all integers $z$, all items~\makebox{$b \in B$} with~\makebox{$b[k] = z$} end up in the same block $B_i$; this is possible since there are at most $\Order(|B| / g)$ such elements given that~$z$ is not a pivot. Thus, the strict inequalities $\max_k(B_j) < \min_k(B_{j+1})$ demanded by Property~(2) apply. By inserting the heavy blocks at the appropriate places into this order, we have altogether satisfied Property~(2). This completes the description of the partition.

To continue, let us call a pair $(i, j) \in [g]^2$ \emph{partially relevant} if $\min_k(A_i) + \min_k(B_j) \leq t$ and \emph{totally relevant} if $\max_k(A_i) + \max_k(B_j) \leq t$. Note that each totally relevant pair is also partially relevant. For each partially relevant pair $(i, j)$ we compute the set $C_{i, j} = (A_i + B_j) \cap [0\ldots t]^{k-1} \times [0\ldots \infty]^{d-k+1}$ recursively (i.e., with parameter $k - 1$), and add all relevant points to the output (i.e., discarding all vectors $c \in C_{i, j}$ with $c[k] > t$). The correctness of this approach is clear, as all pairs $(i, j)$ that are not partially relevant cannot contribute any points to the output.

For the running time analysis, we call the set $\mathcal C_\Delta = \set{ (i, j) \in [g]^2 : i - j = \Delta}$ a ``chain''; note that the only non-empty chains are $\mathcal C_{-g+1}, \dots, \mathcal C_{g-1}$. Let $T_k(|A|, |B|, |C|)$ denote the running time of the algorithm (with parameter $k$). Then, up to some linear-time bookkeeping we can bound the running time by summing over all partially relevant pairs $(i,j)$ as
\begin{align*}
    &\sum_{\substack{(i, j) \in [g]^2\\\text{part.\ relevant}}} T_{k-1}(|A_i|, |B_j|, |C_{i, j}|) \\
    &\qquad= \sum_{\Delta = -g+1}^{g-1} \sum_{\substack{(i, j) \in \mathcal C_\Delta\\\text{part.\ relevant}}} T_{k-1}(|A_i|, |B_j|, |C_{i, j}|)
\intertext{Now observe the following two facts: First, note that in each chain there is at most one pair $(i, j) = (i, i + \Delta)$ that is partially but not totally relevant, since such a pair satisfies $\min_k(A_i) + \min_k(B_{i+\Delta}) \leq t$ and $\max_k(A_i) + \max_k(B_{i + \Delta}) > t$. Note further that $i + \Delta$ cannot be heavy as this would contradict Property~(4)~(ii), and thus $i + \Delta$ is light. For this one pair we will apply the trivial bound~\smash{$T_{k-1}(|A_i|, |B_j|, |C_{i, j}|) \leq \widetilde\Order(|A_i| \, |B_j|) \leq \widetilde\Order(\frac{|A| \, |B|}{g^2})$}, where the last inequality is due to Properties~(3) and~(4)~(i). It follows that we can further bound the running time by:}
    &\qquad= \sum_{\Delta = -g+1}^{g-1} \parens*{ \widetilde\Order\parens*{\frac{|A| \, |B|}{g^2}} + \sum_{\substack{(i, j) \in \mathcal C_\Delta\\\text{tot.\ relevant}}} T_{k-1}(|A_i|, |B_j|, |C_{i, j}|)} \\
    &\qquad= \sum_{\Delta = -g+1}^{g-1} \widetilde\Order\parens*{ \frac{|A| \, |B|}{g^2} + \sum_{\substack{(i, j) \in \mathcal C_\Delta\\\text{tot.\ relevant}}} (|A_i| \, |B_j|)^{1-\frac{1}{k}} |C_{i, j}|^{\frac{1}{k}} + |A_i| + |B_j|} \\
    &\qquad= \sum_{\Delta = -g+1}^{g-1} \widetilde\Order\parens*{ \frac{|A| \, |B|}{g^2} + \parens*{\sum_{\substack{(i, j) \in \mathcal C_\Delta\\\text{tot.\ relevant}}} |A_i| \, |B_j|}^{1-\frac{1}{k}} \cdot \parens*{\sum_{\substack{(i, j) \in \mathcal C_\Delta\\\text{tot.\ relevant}}} |C_{i, j}|}^{\frac{1}{k}} + |A| + |B|},
\intertext{where in the last two steps we have applied the induction hypothesis and Hölder's inequality. Next, observe that all sets $C_{i, j}$ along a common chain $\mathcal C_\Delta$ do not intersect. Indeed, this follows from $\max_k(C_{i, j}) = \max_k(A_i) + \max_k(B_j) < \min_k(A_{i+1}) + \min_k(B_{j+1}) = \min_j(C_{i+1, j+1})$, where we have applied Properties~(1) and~(2) (crucially using the strict inequalities). In particular, summing over all totally relevant pairs $(i, j)$ in a common chain gives $\sum_{(i, j)} |C_{i, j}| \leq |C|$. From this we can finally bound:}
    &\qquad= \sum_{\Delta = -g+1}^{g-1} \widetilde\Order\parens*{ \frac{|A| \, |B|}{g^2} + \parens*{\sum_{\substack{(i, j) \in \mathcal C_\Delta\\\text{tot.\ relevant}}} |A_i| \, |B_j|}^{1-\frac{1}{k}} \cdot |C|^{\frac{1}{k}} + |A| + |B|} \\
    &\qquad= \sum_{\Delta = -g+1}^{g-1} \widetilde\Order\parens*{ \frac{|A| \, |B|}{g^2} + \parens*{\frac{|A|}{g} \cdot \sum_{\substack{(i, j) \in \mathcal C_\Delta\\\text{tot.\ relevant}}} |B_j|}^{1-\frac{1}{k}} \cdot |C|^{\frac{1}{k}} + |A| + |B|}, \\
    &\qquad= \sum_{\Delta = -g+1}^{g-1} \widetilde\Order\parens*{ \frac{|A| \, |B|}{g^2} + \parens*{\frac{|A| \, |B|}{g}}^{1-\frac{1}{k}} \cdot |C|^{\frac{1}{k}} + |A| + |B|} \\
    &\qquad= \widetilde\Order\parens*{ \frac{|A| \, |B|}{g} + g^{\frac{1}{k}} \cdot (|A| \, |B|)^{1 - \frac{1}{k}} \cdot |C|^{\frac{1}{k}} + g |A| + g |B|}.
\end{align*}
By choosing $g = (|A| \, |B| / |C|)^{\frac{1}{k+1}}$ the running time becomes $\widetilde\Order(|A| + |B| + (|A|\, |B|)^{1 - \frac{1}{k+1}} |C|^{\frac{1}{k+1}})$ as claimed. This is easy to verify for the first terms. For the latter two terms note that
\begin{equation*}
    g |A| \leq g |A|^{1 - \frac{2}{k+1}} \cdot |C|^{\frac{2}{k+1}} = |A|^{1 - \frac{1}{k+1}} |B|^{\frac{1}{k+1}} |C|^{\frac{1}{k+1}} \leq (|A| \, |B|)^{1 - \frac{1}{k+1}} |C|^{\frac{1}{k+1}}.
\end{equation*}

As a subtlety, we remark that we do not know $|C|$ a priori, and thus cannot simply choose~$g$. However, we can run the algorithm with exponentially increasing \emph{guesses}~$s \gets 2, 4, 8, \dots$ of $|C|$. For each guess we set~\smash{$g = (|A| \, |B| / s)^{\frac{1}{k+1}}$}, let the algorithm run for~\smash{$\widetilde\Order(|A| + |B| + (|A|\, |B|)^{1 - \frac{1}{k+1}} s^{\frac{1}{k+1}})$} steps, and if the time budget is exceeded we interrupt and continue with the next guess. Otherwise, we return the set~$C$ computed by the algorithm. The total running time is bounded by the geometric sum
\begin{equation*}
    \sum_{\ell \leq \log |C| + 1} \widetilde\Order(|A| + |B| + (|A|\, |B|)^{1 - \frac{1}{k+1}} (2^\ell)^{\frac{1}{k+1}}) = \widetilde\Order(|A| + |B| + (|A|\, |B|)^{1 - \frac{1}{k+1}} |C|^{\frac{1}{k+1}}),
\end{equation*}
which only worsens the running time by a logarithmic factor.

We finally comment on the logarithmic factors. The initial use of the orthogonal range searching data structure as well as guessing $g$ lead to an overhead of $(\log(|A| \, |B|))^{\Order(d)}$. The remaining logarithmic factors in the algorithm are due to the sparse convolutions at the base level of the recursion. By \cref{lem:sparse-conv}, this step can be implemented by a Las Vegas algorithm with a running time overhead of $\log(|A| |B|)$, or by a deterministic algorithm with an overhead of $\polylog(t)$.
\end{proof}

From this theorem we can now easily derive the high-dimensional ``large items'' lemma:

\begin{lemma}[Large Items] \label{lem:large-items-high-dim}
Let $d \geq 1$ be a constant, let $\gamma > 0$, let $X_L \subseteq [0\ldots t]^d \setminus [0\ldots \gamma t]^d$ be a multiset and let $Z \subseteq [0\ldots t]^d$. We can compute $Z' = (Z + \SSS(X_L, t)) \cap [0\ldots t]^d$ in time \smash{$\widetilde\Order(|Z'| \cdot |X_L|^{1-\frac{1}{d+1}} \cdot \poly(\gamma^{-1}))$} (by a Las Vegas algorithm).
\end{lemma}
\begin{proof}
Let $m = \ceil{d / \gamma}$. Suppose that we have a family of hash functions $\mathcal H \subseteq \set{ h : X \to [K] }$ (for some parameter $K$) with the following property: For any distinct elements $x_1, \dots, x_m \in X_L$, there is a hash function $h \in \mathcal H$ such that $h(x_1), \dots, h(x_m)$ are pairwise distinct (we say that $x_1, \dots, x_m$ are \emph{hashed perfectly}). We will first show how we can use this family to prove the lemma statement, and later show how to deterministically construct $\mathcal H$.

To compute $Z'$, we enumerate all hash functions $h \in \mathcal H$ and define $X_{h, i} = \set{x \in X : h(x) = i}$; note that this induces a partition $X_L = X_{h, 1} \sqcup \dots \sqcup X_{h, K}$ into $K$ parts. For each function $h \in \mathcal H$ we compute the set
\begin{equation*}
    Z'_h = (Z + (X_{h, 1} \cup \set{0}) + \dots + (X_{h, 1} \cup \set{0})) \cap [0\ldots t]^d,
\end{equation*}
by $K$ applications of \cref{thm:prefix-sumset-high-dim}. Then we return $\bigcup_{h \in \mathcal H} Z'_h$. 

It is easy to verify that $\bigcup_{h \in \mathcal H} Z'_h \subseteq Z$, and it remains to verify the converse direction. Take any element $z \in Z$, which can expressed as $z' = z + x_1 + \dots + x_\ell$ where $z \in Z$, $x_1, \dots, x_\ell \in X_L$. Note that $\ell \leq d / \gamma \leq m$, as each vector $x_i$ is at least $\gamma t$ in at least one of the $d$ coordinates, whereas~\makebox{$z' \in [0\ldots t]^d$}. Thus, there is a hash function $h \in \mathcal H$ that perfectly hashes $x_1, \dots, x_\ell$ and it follows from the construction that $z' \in Z'_h$.

It remains to construct $\mathcal H$. We employ the construction behind Reed-Solomon codes. Specifically, let $p = K$ be a prime. We identify $[K]$ arbitrarily with the finite field $\Field_p$. Let $D = \ceil{\log_p |X_L|}$. Then we can uniquely identify each item~\makebox{$x \in X_L$} with a degree-$D$ polynomial $P_x$ over $\Field_p$. Finally, take $\mathcal H = \set{h_0, \dots, h_{p-1}}$, where the hash functions $h_i : X \to \Field_p$ are defined by $h_i(x) = P_x(i)$. This completes the construction of $\mathcal H$, and in the following we argue that this family is as desired. First consider an arbitrary pair of distinct items $x_1, x_2 \in X_L$. Since $P_{x_1}$ and $P_{x_2}$ are distinct degree-$D$ polynomials, they can be equal on at most $D$ points. Thus, there are at most $D$ hash functions $h_i$ with $h_i(x_1) = h_i(x_2)$. Similarly, for any distinct items $x_1, \dots, x_m$ there are at most $\binom{m}{2} \cdot D$ hash functions that do not perfectly hash $x_1, \dots, x_m$. Choosing $p \in [m^2 \cdot \ceil{\log |X_L|}, 2m^2 \cdot \ceil{\log |X_L|}]$ it holds that $\binom{m}{2} \cdot D < p = |\mathcal H|$, hence there is at least one good hash function and the claim follows.

It remains to analyze the running time. Operations over the finite field $\Field_p$ can easily be implemented in time $\poly(p) = \poly(m \log |X_L|)$, and therefore the construction of the hash family~$\mathcal H$ takes time $\widetilde\Order(|X_L| \poly(m)) = \widetilde\Order(|X_L| \poly(\gamma^{-1}))$. The running time of the algorithm is then dominated by the $|\mathcal H| \cdot K = \poly(m \log |X_L|) = \poly(\gamma^{-1} \log |X_L|)$ calls to \cref{thm:prefix-sumset-high-dim}. 
Each such call is with sets $A, C \subseteq Z'$ and $B \subseteq X_L \cup \set{0}$, and thus runs in time~\smash{$\Order(|Z'| \cdot |X_L|^{1 - \frac{1}{d+1}} \polylog(t))$}, if we use the deterministic algorithm of \cref{thm:prefix-sumset-high-dim}. However, here we want to avoid the $\polylog(t)$ factor, and thus use the Las Vegas algorithm of \cref{thm:prefix-sumset-high-dim}, so that each call runs in time ~\smash{$\widetilde\Order(|Z'| \cdot |X_L|^{1 - \frac{1}{d+1}})$}. Since the rest of the algorithm is deterministic, in total we obtain a Las Vegas algorithm. (Note that it indeed was necessary to use a deterministic hash family $\mathcal H$ to obtain this result, as a random hash family would only yield a Monte Carlo randomized algorithm.)
\end{proof}

\subsection{Color Coding} \label{sec:high-dim:sec:color-coding}
The second step is to derive the high-dimensional analogue of the color coding lemma for small items:

\begin{lemma}[Small Items] \label{lem:small-items-high-dim}
Let $d \geq 1$ be a constant, let $K \geq 1$ and let $X \subseteq [0\ldots \frac{t}{K^5}]^d$ be a multiset. Let $X = X_1 \sqcup \dots \sqcup X_K$ denote a uniformly random partition into $K$ parts. Then with probability at least $1 - \exp(-K) \cdot |\SSS(X, t)| \cdot d$ it holds that
\begin{equation*}
    \SSS(X, t) = \Big( \sum_{i=1}^K \SSS(X_i, (1 + \tfrac{1}{K}) \cdot \tfrac{t}{K}) \Big) \cap [0\ldots t].
\end{equation*}
\end{lemma}

For the proof of \cref{lem:small-items-high-dim} we rely on Hoeffding's inequality:

\begin{lemma}[Hoeffding's Inequality]
Let $Z_1, \dots, Z_n$ be independent random variables such that~$Z_i$ takes values in $[a_i, b_i]$. Then for $Z = Z_1 + \dots + Z_n$ it holds that
\begin{equation*}
    \Pr\parens*{Z - \Ex[Z] \geq \lambda} \leq \exp\parens*{-\frac{2\lambda^2}{\sum_{i=1}^n (b_i - a_i)^2}}.
\end{equation*}
\end{lemma}

\begin{proof}[Proof of \cref{lem:small-items-high-dim}]
Consider any subset sum $z \in \SSS(X, t)$ expressed as $z = x_1 + \dots + x_m$; we show that with good probability~\smash{$z \in \sum_{i=1}^K \SSS(X_i, (1 + \tfrac{1}{K}) \cdot \tfrac{t}{K})$}. Fix an arbitrary part $X_i$ and a coordinate~\makebox{$\ell \in [d]$}. Let~\makebox{$Z_1, \dots, Z_m$} be random variables such that $Z_j = x_j[\ell]$ if $x_j$ is placed in the part $X_i$, and~\makebox{$Z_j = 0$} otherwise. Writing $Z = Z_1 + \dots + Z_m$ we can bound the expectation of $Z$ by
\begin{equation*}
    \Ex[Z] = \Ex[Z_1] + \dots + \Ex[Z_m] \leq \frac{x_1[\ell]}{K} + \dots + \frac{x_m[\ell]}{K} = \frac{z[\ell]}{K} \leq \frac{t}{K}.
\end{equation*}
Using that the random variable $Z_j$ takes values in $[0\ldots x_j[\ell]]$ (where $x_j[\ell] \leq t / K^5$ and $\sum_j x_j[\ell] \leq t$), from Hoeffding's inequality it follows that
\begin{align*}
    &\Pr\parens*{Z > \parens*{1 + \frac{1}{K}} \cdot \frac{t}{K}} \\
    &\qquad\leq \Pr\parens*{Z - \Ex[Z] > \frac{t}{K^2}} \\
    &\qquad\leq \exp\parens*{-\frac{2(t / K^2)^2}{\sum_{j=1}^m x_j[\ell]^2}} \\
    &\qquad\leq \exp\parens*{-\frac{2(t / K^2)^2}{t^2 / K^5}} \\
    &\qquad= \vphantom{\bigg(}\exp\parens*{-2K}.
\end{align*}
That is, with probability at least $1 - \exp(-2K)$ the contribution of the subset $z$ to the $i$-th bucket in the $\ell$-th coordinate does not exceed the threshold $(1 + \frac{1}{K}) \cdot \frac{t}{K}$. Taking a union bound over the at most $|\SSS(X, t)|$ subset sums, the $K$ buckets and the $d$ dimensions, the total error probability indeed becomes $\exp(-2K) \cdot |\SSS(X, t)| \cdot K \cdot d \leq \exp(-K) \cdot |\SSS(X, t)| \cdot d$.
\end{proof}

\subsection{The Full Algorithm} \label{sec:high-dim:sec:full}
In this section we summarize the algorithm for high-dimensional Subset Sum; see \cref{alg:subset-sum-high-dim}. As before, the algorithm expects as a third input a number $s$ satisfying that $|\SSS(X, t)| \leq s \leq \poly(|\SSS(X, t)|)$. In the following \cref{sec:high-dim:sec:scaling} we finally remove this assumption.

\begin{algorithm}[t]
\caption{The fast algorithm for high-dimensional Subset Sum (see \cref{thm:subset-sum-high-dim})} \label{alg:subset-sum-high-dim}
\begin{algorithmic}[1]
    \Procedure{FastSubsetSum}{$X, t, s$}
        \If { $|X| \leq 1$ or $t \leq 1$} \label{alg:subset-sum-high-dim:line:trivial}
            \State Solve the instance trivially
        \EndIf
        \State Let $Z \gets \set{0}$
        \State Let $K \gets \ceil{100 \log s + \log d}$
        \State Let $X_S \gets X \cap [0 \ldots \frac{t}{K^5}]^d$ and $X_L \gets X \setminus X_S$

        \medskip
        \State\emph{(Small items)}
        \State Partition $X_S$ randomly into multisets $X_S = X_1 \sqcup \dots \sqcup X_K$
        \For {$k \in [K]$}
            \State Recursively compute $Z_k \gets \textsc{FastSubsetSum}(X_k, (1+\frac{1}{K}) \cdot \frac{t}{K}, s)$
            \State Compute $Z \gets Z + Z_k$ by \cref{lem:sparse-conv} \label{alg:subset-sum-high-dim:line:sumset}
        \EndFor

        \medskip
        \State\emph{(Large items)}
        \State\Return $(Z + \mathcal{S}(X_L, t)) \cap [0 \ldots t]^d$ computed by \cref{lem:large-items-high-dim} (with~\smash{$\gamma = \frac{1}{K^5}$}) \label{alg:subset-sum-high-dim:line:large}
    \EndProcedure
\end{algorithmic}
\end{algorithm}

For technical reasons, we start with the following lemma that establishes a crude bound on the total number of recursive calls. We say that a call to $\textsc{FastSubsetSum}(X, t, s)$ is \emph{trivial} if $X$ is empty.

\begin{lemma}[Number of Recursive Calls] \label{lem:subset-sum-rec-calls}
With probability at least $1 - 1/|X|^{40}$, the recursion depth of \cref{alg:subset-sum-high-dim} is at most $400 \log |X|$, and in that case the number of recursive calls is at most $400 |X| \log |X|$.
\end{lemma}
\begin{proof}
We say that a call to the algorithm is \emph{good} if it selects a partition $X_S = X_1 \sqcup \dots \sqcup X_K$ in which all parts have size at most $|X_S| / 2$. We claim that each call is good with probability at least $\frac{1}{2}$. Indeed, a call is bad only if there are at least~\smash{$\binom{|X_S| / 2}{2}$} \emph{collisions}, i.e., pairs~\makebox{$x, x' \in X_S$} of distinct elements that have fallen into the same bucket $X_i$. But the expected number of collisions is~\smash{$\binom{|X_S|}{2} / K \leq \binom{|X_S|}{2} / 100$} and therefore, by applying Markov's inequality, each call is bad with probability at most $(\binom{|X_S|}{2} / 100) / \binom{|X_S| / 2}{2} \leq \frac{1}{2}$. 
By Chernoff's bound, for any length-$\ell$ path in the recursion tree starting at the root, at least $\ell/4$ of the recursive calls choose a good partition with probability at least $1 - \exp(-\ell/8)$. Therefore, with probability at least $1 - 1/|X|^{50}$ any particular path of length $\ell := 400 \log |X|$ contains at least $100 \log |X| \ge \log |X|$ good partitions and thus reaches a leaf before reaching the end of the path. In particular, by a union bound over all $\le |X|$ leaves, the recursion depth is bounded by $400 \log |X|$ with probability at least $1 - 1/|X|^{40}$. The claim follows. 
\end{proof}

\begin{lemma}[Correctness of \cref{alg:subset-sum-high-dim}] \label{lem:subset-sum-correctness}
Let $X \subseteq [0 \ldots t]^d$ be a multiset, and assume $s \geq |\SSS(X, t)|$. Then \cref{alg:subset-sum-high-dim} correctly computes $\SSS(X, t)$ with high probability.
\end{lemma}
\begin{proof}
It is easy to verify that the algorithm never reports wrong subset sums, but it may omit subset sums in its output. To bound this error event, focus on any recursive call and consider any subset sum $z \in \SSS(X, t)$. Write $z = z_S + z_L$ for the contributions from $X_S$ and $X_L$, respectively. The computation of the large items in \cref{alg:subset-sum:line:large} is Las Vegas, thus the error event only involves the small items. Specifically, the algorithm errs if $z_S \not\in \sum_{k=1}^K \SSS(X_k, (1 + \frac{1}{K}) \cdot \frac{t}{K})$, which happens with probability at most $\exp(-K) \cdot |\SSS(X, t)| \cdot d$ by \cref{lem:small-items}. Taking a union bound over the at most~$400 |X| \log |X|$ recursive calls (by the previous \cref{lem:subset-sum-rec-calls}), the error probability is at most
\begin{equation*}
    \exp(-K) \cdot d \cdot |\SSS(X, t)|^2 \cdot 400|X| \log |X| \leq s^{-100} \cdot d^{-1} \cdot d \cdot |\SSS(X, t)|^2 \cdot 400|X| \log |X| \leq O(|\SSS(X, t)|^{-90}),
\end{equation*}
using the assumptions $|X| \leq |\SSS(X, t)| \leq s$. (Of course, the constant $90$ can be made larger by adjusting the constant in the definition of $K$.)
\end{proof}

\begin{lemma}[Running Time of \cref{alg:subset-sum-high-dim}] \label{lem:subset-sum-time}
Let $X \subseteq [0 \ldots t]^d$ be a multiset, and let $s \geq |\SSS(X, t)|$. Then \cref{alg:subset-sum-high-dim} runs in time~\smash{$|\SSS(X, t)| \cdot |X|^{1-\frac{1}{d+1}} \cdot (\log s)^{\Order(d)}$} with high probability.
\end{lemma}
\begin{proof}
Consider first the top-level execution of \cref{alg:subset-sum} and ignore the recursive calls. The construction of~$X_S$ and~$X_L$, and the partition of $X_S$ into smaller subsets runs in linear time~$\Order(|X|)$. There are~$K$ calls of \cref{alg:subset-sum:line:sumset} each of which runs in time
\begin{equation*}
    \widetilde\Order(|Z_1 + \dots + Z_k|) = \widetilde\Order\parens*{\abs*{\sum_{k \in [K]} \SSS(X_k, (1 + \tfrac{1}{K}) \cdot \tfrac{t}{K})}}
\end{equation*}
Pick $A_k = \SSS(X_k, (1 + \frac{1}{K}) \cdot \frac{t}{K})$ and observe that $B_k = \sum_{i \neq k} A_i \subseteq \SSS(X, t)$ (as $(K - 1) \cdot (1 + \frac{1}{K}) \cdot \frac{t}{K} \leq t$). Then from \cref{thm:submultiplicativity} it follows that
\begin{equation*}
    |A_1 + \dots + A_K| \leq \parens*{\prod_{k=1}^K |B_k|}^{\frac{1}{K-1}} \leq |\SSS(X, t)|^{\frac{K}{K-1}} \leq |\SSS(X, t)|^{1 + \frac{1}{K-1}}.
\end{equation*}
As we choose $K \geq \log s \geq \log |\SSS(X, t)|$, the $K$ repetitions of \cref{alg:subset-sum-high-dim:line:sumset} amount to~\smash{$|\SSS(X, t)| \cdot \polylog(s)$} time. Calling \cref{lem:large-items-high-dim} in \cref{alg:subset-sum-high-dim:line:large} takes time~\smash{$|\SSS(X, t)| \cdot |X|^{1-\frac{1}{d+1}} \cdot (\log s)^{\Order(d)}$}.

Finally, let us take the cost of the recursive calls into account. Throughout we condition on the event that the recursion depth reaches at most $\Order(\log |X|)$ (which happens with high probability by \cref{lem:subset-sum-rec-calls}). Then the computation can be modeled by a recursion tree with fan-out at most~$K$, depth at most $\Order(\log |X|)$, and at most $\Order(|X| \log |X|)$ nodes in total (by \cref{lem:subset-sum-rec-calls}). Focus on any (non-root) level~\makebox{$\ell \geq 1$} in this tree with $m \leq \min\set{|X|, K^\ell}$ nodes. This level induces a partition of~$X$ into parts~\makebox{$X_1 \sqcup \dots \sqcup X_m$}. By the previous paragraph, the total running time spend on the current level~$\ell$ is 
\begin{align*}
    &\sum_{i \in [m]} \abs*{\SSS(X_i, (1 + \tfrac{1}{K})^\ell \cdot \tfrac{t}{K^\ell})}^{1+\frac{1}{K-1}} \cdot |X|^{1-\frac{1}{d+1}} \cdot (\log s)^{\Order(d)} \\
    &\qquad\leq \sum_{i \in [m]} \abs*{\SSS(X_i, (1 + \tfrac{1}{K})^\ell \cdot \tfrac{t}{K^\ell})} \cdot |\SSS(X, t)|^{\frac{1}{K-1}} \cdot |X|^{1-\frac{1}{d+1}} \cdot (\log s)^{\Order(d)} \\
    &\qquad\leq \sum_{i \in [m]} \abs*{\SSS(X_i, (1 + \tfrac{1}{K})^\ell \cdot \tfrac{t}{K^\ell})} \cdot |X|^{1-\frac{1}{d+1}} \cdot (\log s)^{\Order(d)}.
\end{align*}
Let~\smash{$C_i = \SSS(X_i, (1 + \tfrac{1}{K})^\ell \cdot \tfrac{t}{K^\ell})$}. To bound the remaining sum~\smash{$\sum_{i \in [m]} |C_i|$}, first note that the maximum integer in $\sum_{i \in [m]} C_i$ has size
\begin{equation*}
    m \cdot \parens*{1 + \frac{1}{k}}^\ell \cdot \frac{t}{K^\ell} \leq t \cdot \parens*{1 + \frac{1}{100 \log |X|}}^{\Order(\log |X|)} \cdot t \leq c t,
\end{equation*}
for some (integer) constant $c \geq 1$. Therefore, by partitioning $[m]$ into $c$ equal-sized chunks $I_1, \dots, I_c$ we can guarantee that $\sum_{i \in I_j} C_i \subseteq \SSS(X, t)$. Hence, by applying \cref{lem:sumset-lower-bound} we have that
\begin{align*}
    &\sum_{i \in [m]} |C_i| = \sum_{j=1}^c \sum_{i \in I_j} |C_i| \leq \sum_{j=1}^c \parens*{|I_j| + \abs*{\sum_{i \in I_j} C_i}} = \Order(|\SSS(X, t)|).
\end{align*}
In summary, each level of recursion runs in total time~\smash{$|\SSS(X, t)| \cdot |X|^{1-\frac{1}{d+1}} \cdot \polylog(s)$}, and summing over the $\log |X| \leq \log s$ levels worsens the running time only by a logarithmic factor. 
\end{proof}

\subsection{Approximating the Number of Subset Sums} \label{sec:high-dim:sec:scaling}
Finally, we remove the assumption that the algorithm requires an estimate $s$ of $|\SSS(X, t)|$ as an input, and thereby complete the proof of \cref{thm:subset-sum-high-dim}. Recall that we could easily set $s = t$ (which is a legal choice) to obtain an algorithm with an overhead of $\polylog(t)$ rather than $\polylog |\SSS(X, t)|$.

\begin{proof}[Proof of \cref{thm:subset-sum-high-dim}]
Provided that $s \geq |\SSS(X, t)|$ we have already established in \cref{lem:subset-sum-correctness,lem:subset-sum-time} that \cref{alg:subset-sum} is correct and runs in time~\smash{$|\SSS(X, t)| \cdot |X|^{1-\frac{1}{d+1}} \cdot (\log s)^{\Order(d)}$}. We now design a recursive algorithm that computes $\SSS(X, t)$ with high probability, without expecting an approximation~$s$. For constant-sized sets $X$ we can compute the answer trivially, say by Bellman's algorithm. Otherwise, arbitrarily partition $X = X_1 \sqcup X_2$ into sets of size at most $\ceil{|X| / 2}$ and compute the subset sums $\SSS(X_1, \frac{t}{2})$ and $\SSS(X_2, \frac{t}{2})$ recursively. Write $s_1 = |\SSS(X_1, \frac{t}{2})|$ and~\makebox{$s_2 = |\SSS(X_2, \frac{t}{2})|$}, and choose $s = s_1^{4d} \cdot s_2^{4d} \cdot |X|^{4d}$. We then call \cref{alg:subset-sum} on $(X, t, s)$ and report the resulting answer.

For the correctness, we argue that $s \geq |\SSS(X, t)|$. To this end we show the existence of an injection $f : \SSS(X, t) \to X^{4d} \times \SSS(X_1, \frac{t}{2})^{4d} \times \SSS(X_2, \frac{t}{2})^{4d}$. Fix any subset sum~\makebox{$z \in \SSS(X, t)$} expressed as~\makebox{$z = x_1 + \dots + x_m$}; we describe how to choose $f(z)$. First, take aside all vectors $x_i$ with an entry of size larger than $\frac{t}{4}$; note that there can be at most $4d$ such terms. Afterwards, we split the remaining subset sum into~\makebox{$x_{1, 1} + \dots + x_{1, m_1} + x_{2, 1} + \dots + x_{2, m_2}$} where~\makebox{$x_{1, 1}, \dots, x_{1, m_1} \in X_1 \cap [0 \ldots \frac{t}{4}]^d$} and~\makebox{$x_{2, 1}, \dots, x_{2, m_1} \in X_2 \cap [0 \ldots \frac{t}{4}]^d$}. Next, focus on $x_{1, 1} + \dots + x_{1, m_1}$. We can greedily take the maximal prefix sum with value in $[0 \ldots \frac{t}{2}]^d$; since each vector is $\frac{t}{4}$-bounded this prefix sum is at least~$\frac{t}{4}$ in some coordinate. Thus, by repeatedly splitting off these prefix sums we eventually end up splitting the initial subset sum $x_{1, 1} + \dots + x_{1, m_1} \in \SSS(X_1, t)$ into at most $4d$ subset sums contained in~$\SSS(X_1, \frac{t}{2})$. We similarly split $x_{2, 1} + \dots + x_{2, m_2}$. Finally, we choose $f(z)$ to be the $(4d + 4d + 4d)$-tuple obtained from the $4d$ heavy vectors taken aside in the beginning, plus the~\makebox{$4d + 4d$} subdivisions of the remaining subset sums. This mapping $f$ is clearly injective as we can recover $z$ by summing all tuple entries in $f(z)$, and thus it follows that indeed $s = s_1^{4d} \cdot s_2^{4d} \cdot |X|^{4d} \geq |\SSS(X, t)|$.

For the running time observe that we always pick $s = s_1^{4d} \cdot s_2^{4d} \cdot |X|^{4d} \leq |\SSS(X, t)|^{12d}$, and therefore $\log(s) = \Order(\log |\SSS(X, t)|)$. It remains to bound the overhead due to the additional layer of recursion. Again, view the computation as a binary tree with depth $\log |X|$. Each layer $\ell$ of the tree induces a partition of $X$ into $2^\ell \leq |X|$ parts $X = X_1 \sqcup \dots \sqcup X_{2^\ell}$. Each node is associated to an execution of the algorithm that takes time~\smash{$\widetilde\Order(|\SSS(X_i, \frac{t}{2^\ell})| \cdot |X|^{1-\frac{1}{d+1}})$}, and therefore the total running time is bounded by
\begin{gather*}
    \sum_{i=1}^{2^\ell} \widetilde\Order\parens*{|\SSS(X_i, \tfrac{t}{2^\ell})| \cdot |X|^{1-\frac{1}{d+1}}} \\
    \qquad\leq \widetilde\Order\parens*{\parens*{2^\ell + \abs*{\sum_{i=1}^{2^\ell} \SSS(X_i, \tfrac{t}{2^\ell})}} \cdot |X|^{1-\frac{1}{d+1}}} \\
    \qquad= \widetilde\Order\parens*{|\SSS(X, t)| \cdot |X|^{1-\frac{1}{d+1}}},
\end{gather*}
using \cref{lem:sumset-lower-bound} for the first inequality and $2^\ell \le |X| \le |\SSS(X, t)|$ as well as $\sum_{i=1}^{2^\ell} \SSS(X_i, \tfrac{t}{2^\ell}) \subseteq \SSS(X, t)$ for the second.
\end{proof}
\section{Unbounded Subset Sum} \label{sec:unbounded}
Our next goal is to show that the techniques developed in the previous sections similarly apply to the \emph{Unbounded} Subset Sum problem (where we can select each item an unbounded number of times), even by means of deterministic algorithms. Our goal is to prove \cref{thm:unbounded}. In fact, we prove the following generalization to higher dimensions:

\begin{algorithm}[t]
\caption{The fast algorithm for Unbounded Subset Sum (see \cref{thm:unbounded-high-dim})} \label{alg:unbounded}
\begin{algorithmic}[1]
    \Procedure{FastUnboundedSubsetSum}{$X, t$}
        \If{$|X| \leq 1$ or $t \leq \Order(1)$}
            \State Solve the instance naively
        \EndIf
        \State Let $Z \gets \set{0}$
        \State Let $K \gets \ceil{100 \log t + \log d}$
        \State Let $X_S \gets X \cap [0\ldots \frac{t}{K^5}]^d$ and $X_L \gets X \setminus X_S$

        \medskip
        \State\emph{(Small items)}
        \State Recursively compute $Z_0 \gets \textsc{FastUnboundedSubsetSum}(X_S, (1 + \frac{1}{K}) \cdot \frac{t}{K})$
        \RepeatTimes{$K$}
            \State Compute $Z \gets Z + Z_0$ by \cref{lem:sparse-conv} \label{alg:unbounded:line:sumset}
        \EndRepeatTimes

        \medskip
        \State\emph{(Large items)}
        \RepeatTimes{$K^5 \cdot d$}
            \State Compute $Z \gets (Z + X_L) \cap [0\ldots t]^d$ by \cref{thm:prefix-sumset-high-dim}
        \EndRepeatTimes
        \State\Return $Z$
    \EndProcedure
\end{algorithmic}
\end{algorithm}

\begin{theorem}[High-Dimensional Unbounded Subset Sum] \label{thm:unbounded-high-dim}
Let $d \geq 1$ be constant. Given a multiset $X \subseteq [t]^d$ we can compute $\SSS^*(X, t)$ in time~\smash{$\widetilde\Order(|\SSS^*(X, t)| \cdot |X|^{1-\frac{1}{d+1}})$} (by a deterministic algorithm).
\end{theorem}
\begin{proof}
We summarize the pseudocode in \cref{alg:unbounded}. As is typical~\cite{KoiliarisX19,Bringmann17}, our algorithm for Unbounded Subset Sum is significantly simpler compared to Subset Sum as we can avoid color coding altogether (this also allows us to achieve a deterministic algorithm).

The idea is similar to before: We partition the given items $X$ into \emph{small items} $X_S = X \cap [0\ldots \frac{t}{K^5}]^d$ and \emph{large items} $X_L = X \setminus X_S$. That is, the large items are vectors which have at least one entry larger than $\frac{t}{K^5}$.

The first case concerns the small items. We recursively compute the unbounded subset sums~$Z_0$ of~$X_S$ with target $(1 + \frac{1}{K}) \cdot \frac{t}{K}$ (for some parameter $K \geq 1$), and then compute $Z$ as the $K$-fold sumset of $Z_0$ (by applying a deterministic sparse sumset algorithm $K-1$ times, see \cref{lem:sparse-conv}). Before we continue with the algorithm, we show that we thereby compute all unbounded sums up to the given target $t$:

\begin{claim} \label{thm:unbounded-high-dim:clm:small-items}
$Z = \SSS^*(X_S, t)$ (provided that $K > \log |\SSS^*(X_S, t)| + \log d$).
\end{claim}
\begin{proof}
It is clear that $Z \subseteq \SSS^*(X_S, t)$. Focus on the other direction, and consider the multiset $X_S' = \set{2^\ell \cdot x : 0 \leq \ell \leq \log t, x \in X_S}$; it is easy to verify that we can obtain each unbounded subset sum of $X_S$ as a \emph{bounded} subset sum of $X_S'$ (since whenever we take an item $x$, say $m$ times, we can replace these $m$ copies by $\sum_\ell m_\ell \cdot 2^\ell \cdot x$ where $m_i$ denotes the $i$-th bit in the binary representation of $m$), i.e., $\SSS^*(X_S, t) = \SSS(X_S', t)$. Applying \cref{lem:small-items-high-dim} to the set $X_S'$, we obtain that there exists a partition of $X_S' = X_1 \sqcup \dots \sqcup X_K$ such that
\begin{align*}
    \SSS(X_S', t) &= \sum_{i=1}^K \SSS(X_i, (1 + \tfrac{1}{K}) \cdot \tfrac{t}{K}),
\end{align*}
provided that $\exp(-K) \cdot |\SSS(X_S, t)| \cdot d < 1$, as is guaranteed by our choice of $K$. This yields
\begin{align*}
    \SSS(X_S', t) = \sum_{i=1}^K \SSS(X_i, (1 + \tfrac{1}{K}) \cdot \tfrac{t}{K}) &\subseteq \sum_{i=1}^K \SSS^*(X_S, (1 + \tfrac{1}{K}) \cdot \tfrac{t}{K}) = \sum_{i=1}^K Z_0 = Z. \qedhere
\end{align*}
\end{proof}

In light of \cref{thm:unbounded-high-dim:clm:small-items} we have successfully computed all unbounded subset sums involving small items. To deal with the large items, we now update $Z \gets (Z + X_L) \cap [0\ldots t]^d$ for a total of $K^5 d$ times, and then return $Z = \SSS^*(X, t)$. The idea is that each (unbounded) subset sum can involve at most~$K^5 d$ copies of a large item (since each large item has at least one coordinate with size $\geq \frac{t}{K^5}$). The correctness is immediate.

We finally analyze the running time. With each recursive call we decrease the target $t$ to at most $t/2$, and therefore the overhead due to the recursion is at most $\Order(\log t)$. The running time of computing the $K-1 = \Order(\log t)$ sparse sumsets by \cref{lem:sparse-conv} is $|\SSS^*(X, t)| \polylog(t)$. The dominant term is to compute the $K^5 d = \polylog(t)$ prefix-restricted sumsets by \cref{thm:prefix-sumset-high-dim} (for sets $A, C \subseteq \SSS^*(X, t)$ and $B \subseteq X_L$) each running in time~\smash{$|\SSS(X, t)| \cdot |X_L|^{1 - \frac{1}{d+1}} \cdot \polylog(t)$}.
\end{proof}

\bibliographystyle{plainurl}
\bibliography{paper}

\end{document}